\documentclass{birkjour_t2}
 \newtheorem{thm}{Theorem}[section]

 \theoremstyle{definition}
 \newtheorem{defn}[thm]{Definition}
 \theoremstyle{remark}

 \numberwithin{equation}{section}

\begin{document}

%
%
%
%
%
%
%
%
%

\title[Special curves and Hamiltonian systems]{Special curves and Hamiltonian systems from polynomial eigenfunctions of an elliptic operator}

\author[J.G.Escudero]{Juan Garc\'{\i}a Escudero}
\address{%
Madrid, Spain
}
\email{jjgemplubg@gmail.com}


\date{April, 2024}
\dedicatory{}

\begin{abstract}
A family of polynomials linked to the set of the deltoid tangents and its associated algebraic hypersurfaces has been presented in recent years. In this paper we study some related maximising and free plane curves. We also analyse the bifurcations on polynomial Hamiltonian dynamical systems defined from such a family. 

\end{abstract}
\maketitle
\section{Introduction}
\bigskip\par

We study some plane curves and Hamiltonian systems related to a  family of polynomials introduced in the last decade for the construction of algebraic hypersurfaces with many simple singularities  (\cite{esc14, esc16, esc17, esc24} and references within). 
\par
The one-parameter family of degree $d$ polynomials $P_{d}(x,y,\mu)$ is discussed in Section 2, where we also consider the connection of $P_{d}(x,y,\mu)$ with eigenfunctions of a symmetric diffusion operator (see \cite{esc24} for a generalisation to $n$ variables) which corresponds to a Laplace-Beltrami operator for a metric with zero scalar curvature. 
\par
Special curves called free and maximising plane curves are significant in the field of algebraic geometry.  Maximising curves can be used for instance in the construction of algebraic surfaces having the maximal Picard number. A link between maximising curves and free curves has been found in \cite{dim23}. In Section 3 we study some free and maximising curves related to $P_{d}(x,y,\mu)$. In particular we discuss maximising curves of degrees eight and ten. A degree nine curve which is free but not maximising is presented. We also show that a certain maximising decic has a Miyaoka-Kobayashi number close to the maximum possible.
\par
The critical points of $P_{d}(x,y,\mu)$ and some of their properties are presented in Section 4. We use them in Section 5 in order to analyse the bifurcation or catastrophe sets of the associated polynomial Hamiltonian dynamical systems. We find six bifurcation points if $d$ is divisible by three and two bifurcation points otherwise.

\par
\section{A family of degree $d$ polynomials as eigenfunctions of an elliptic operator}
\bigskip\par
We write $P_{d}(x,y,\mu)$ instead of the polynomials $\hat{J}_{d,\tau=\mu}(x,y)$ introduced in \cite{esc17}.  Up to a scaling factor, the polynomials $P_{d}(x,y,\mu)$ are connected to the union of the lines with $d$ orientations linked with the affine Weyl group of the root lattice $\bold{A}_{2}$ and the set of the deltoid tangents  \cite{esc23}. The lines have equations  $L_{d,\mu,s}(x,y):=y+({\rm cos}2\varphi-x){\rm tan}\varphi+{\rm sin}2\varphi=0, \varphi=\frac{6s-1}{6d}-\frac{\mu}{\pi d}, s=- \lfloor \frac{d-2}{2} \rfloor,- \lfloor \frac{d-2}{2} \rfloor+1,..., \lfloor \frac{d+1}{2} \rfloor$ with $(x,y)\in {\Bbb{R}}^2$ and $\mu \in \Bbb{R}$. The basic degree-$d$ polynomials are 
$$P_{d}(x,y,\mu) := \lambda_{d,\mu}  \prod_{s} L_{d,\mu,s}(x,y)$$ 
where $\lambda_{d,\mu} =(-1)^{m}2 d$ if $\mu = (6m-3d-1)\frac{\pi}{6}$ (in this case the line  $L_{d,\mu,s}(x,y)=0$ parallel to the $y$-axis is interpreted as the line $x+1=0$) and  $\lambda_{d,\mu}  = 2{\rm cos}(\mu+\frac{d\pi}{2}+\frac{2\pi}{3})$ if $\mu \neq (6m-3d-1)\frac{\pi}{6}$, $m \in \Bbb{Z}$.
\par
Symmetric diffusion operators defined on an open subset $\Omega  \subset  {\Bbb{R}}^n$ can be written  \cite{bak21}
      \begin{equation} 
  {\mathcal{L}}:= \sum_{ij}g^{ij}(x) \partial_{x_{i}x_{j}}^{2}+\sum_{i}b^{i}(x)\partial_{x_{i}}
    \end{equation}
  such that the symmetric matrix $(g^{ij}(x))$ is non negative at every point $x=(x_{1},x_{2},...x_{n}) \in \Omega$. ${\mathcal{L}}$ is elliptic in the interior of $\Omega$ and an operator of this kind in ${\Bbb{R}}^2$, denoted by ${\mathcal{L}}_{2}$ in \cite{esc17}, has 
      \begin{equation}   
g^{11}(x_{1},x_{2})=-\frac{1}{4}(3x_{1}^2-x_{2}^2-6x_{1}-9), g^{12}(x_{1},x_{2})=g^{21}(x_{1},x_{2})=-\frac{1}{2}(2x_{1}x_{2}+3x_{2}), 
$$  $$
g^{22}(x_{1},x_{2})=-\frac{1}{4}(3x_{2}^2-x_{1}^2+6x_{1}-9), b^{1}(x_{1},x_{2})=-x_{1}, b^{2}(x_{1},x_{2})=-x_{2}
  \end{equation}
  
 The operator  ${\mathcal{L}}_{2}$ is elliptic in the interior of the region $-Q_{\delta}(x_{1},x_{2})\geq 0$, where  $Q_{\delta}(x_{1},x_{2})=(x_{1}^2 + x_{2}^2)^2 - 8 (x_{1}^3 - 3 x_{1} x_{2}^2) + 18 (x_{1}^2 + x_{2}^2) - 27=0$ is the equation of a Steiner hypocycloid or deltoid $\delta$. 
 \par
  For each $\mu$, $J_{d}(x_{1},x_{2},\mu):=P_{d}(x_{1},x_{2},\mu)-2 {\rm cos}3 \mu$ is an eigenfunction of ${\mathcal{L}}_{2}$ \cite{esc17}: 
       \begin{equation}  
  {\mathcal{L}}_{2}J_{d}(x_{1},x_{2},\mu)=-d^{2}J_{d}(x_{1},x_{2},\mu)
      \end{equation} 
  \par
${\mathcal{L}}_{2}$ has the form of the Laplace-Beltrami operator $\frac{1}{\sqrt{|g|}}\sum_{i,j}\partial_{x_{i}}(\sqrt{|g|}g^{ij}\partial_{x_{j}})$
where $g^{ij}$ is the (co)metric or inverse metric given in Eq. (2.2). The metric is  $g_{ij}=(g^{ij})^{-1}$ and $|g|=\det g_{ij}=-\frac{3}{16}Q_{\delta}$.
By computing the Christoffel symbols  $\Gamma^{m}_{jk}:=\frac{1}{2}\sum_{i}g^{im}(\partial_{x_{k}}g_{ij}+\partial_{x_{j}}g_{ki}-\partial_{x_{i}}g_{jk})$
and the Riemann curvature tensor $R^{m}_{ijk}:=\partial_{x_{j}}\Gamma^{m}_{ki}-\partial_{x_{k}}\Gamma^{m}_{ij}+\sum_{s}\Gamma^{s}_{ki}\Gamma^{m}_{js}-\sum_{s}\Gamma^{s}_{ij}\Gamma^{m}_{ks}$ we obtain $R^{m}_{ijk}=0$, hence also both the Ricci tensor $R_{ik}:=\sum_{j}R^{j}_{ijk}$ and the scalar curvature  
$R:=\sum_{i,k}g^{ik}R_{ik}$ are zero. 
\par
The operator ${\mathcal{L}}_{2}$ can be interpreted as a projection of the Euclidean Laplacian $\Delta_{{\Bbb{R}}^2}$. We use the operator
       \begin{equation} 
{\mathcal{G}}_{{\mathcal{L}}}(f,g):=\frac{1}{2}({\mathcal{L}}(fg)-f{\mathcal{L}}g-g{\mathcal{L}}f)
       \end{equation} 
       and we see that $g^{jk}={\mathcal{G}}_{{\mathcal{L}}}(x_{j},x_{k}), b^{j}={\mathcal{L}}x_{j}$. 
\par       
               In terms of $z=x_{1}+\sqrt{-1} x_{2}$ and $\bar{z}=x_{1}-\sqrt{-1} x_{2}$ the operator ${\mathcal{L}}_{2}$ has a simpler form, denoted by  ${\mathcal{L}}_{2C}$, with \cite{esc17} 
        \begin{equation}       
       {\mathcal{G}}_{2C}(z,z)=-z^2+3\bar{z},  {\mathcal{G}}_{2C}(z,\bar{z})={\mathcal{G}}_{2C}(\bar{z},z)=9-z\bar{z}, {\mathcal{G}}_{2C}(\bar{z},\bar{z})=-\bar{z}^2+3z, {\mathcal{L}}_{2C}z=-z, {\mathcal{L}}_{2C}\bar{z}=-\bar{z} 
       \end{equation}
       \par   
       Under a change of variables $y=(y_{k}(x))_{k=1,2,...}$ we have in general 
                 \begin{equation} 
{\mathcal{L}}\phi (y)= \sum_{ij}{\mathcal{G}}_{{\mathcal{L}}}(y_{i}, y_{j}) \partial_{y_{i}y_{j}}^{2}\phi (y)+\sum_{i}({\mathcal{L}}y_{i})\partial_{y_{i}}\phi (y)
       \end{equation}  

\par\noindent
  The so-called generalised cosine (see \cite{esc16} and references therein)
          \begin{equation}  
  h(u,v):=e^ {-2\pi u \sqrt{-1}}+e^ {-2\pi v \sqrt{-1}}+e^ {2\pi (u+v) \sqrt{-1}}
         \end{equation}
 can be used to make a change of variables which shows that ${\mathcal{L}}_{2}$ is a projection of $\Delta_{{\Bbb{R}}^2}$. We identify the sets ${\Bbb{R}}^2$ and ${\Bbb{C}}$ and we make the substitution $ 2\pi u=-\omega \cdot z, 2\pi v=-{\bar{\omega}} \cdot z$ in Eq.(2.7), where $z=x_{1}+\sqrt{-1} x_{2} \in {\Bbb{C}}$,  $z_{1}\cdot z_{2}={\rm Re}(z_{1}\bar{z}_{2})$ and $1, \omega, \bar{\omega}$ are the cubic roots of unity. The variable change $Z: {\Bbb{R}}^2 \rightarrow {\rm Int} \delta$, where ${\rm Int} \delta$ denotes the interior of the region whose border is the deltoid $\delta$, is
        \begin{equation}  
Z(z)=e^ {(1 \cdot z)\sqrt{-1}}+ e^ {(\omega \cdot z)\sqrt{-1}}+ e^ {({\bar{\omega}} \cdot z)\sqrt{-1}}
         \end{equation}
        It  can been shown that $\Delta_{{\Bbb{R}}^2} Z=-Z, \Delta_{{\Bbb{R}}^2} \bar{Z}=-\bar{Z}$ and ${\mathcal{G}}_{\Delta_{{\Bbb{R}}^2}}={\mathcal{G}}_{2C}$  where  ${\mathcal{G}}_{2C}$ is given by Eq.(2.5) for $z=Z, \bar{z}=\bar{Z}$, hence
              \begin{equation}  
        {\mathcal{L}}_{2C}f(Z, \bar{Z}) =\Delta_{{\Bbb{R}}^2} f(Z, \bar{Z})
             \end{equation}    
 \par\noindent            
   The Jacobian $|\frac{\partial(x_{1}^{{'}},x_{2}^{{'}})}{\partial(x_{1},x_{2})}|=4 \sqrt{3} ({\rm cos} (3 x_{1}) - {\rm cos} (\sqrt{3} x_{2})) {\rm sin} (\sqrt{3} x_{2})$, where $x_{1}^{{'}}={\rm Re} Z(z),x_{2}^{{'}}={\rm Im} Z(z)$, is zero in three families of lines, with equations $3x_{1}= \pm \sqrt{3} x_{2}+2 k \pi, \sqrt{3} x_{2}= l \pi, k,l\in {\Bbb{Z}}$, forming a periodic triangular tiling in ${\Bbb{R}}^2$. $Z$ is a one-to-one map from the border of a triangle, which is a fundamental region of the affine Weyl group $\bold{A}_{2}$, to $\delta$. Also $Z$ is a one-to-one map from the interior of the triangle onto ${\rm Int} \delta$ and maps ${\Bbb{R}}^2$ onto ${\rm Int} \delta$.

   \section{Some free and maximising curves}

In this Section we first recall some basic facts on free and maximising
curves in $\mathbb{P}^2$ and then we present some examples based on the polynomials $P_{d}(x,y,\mu)$. We also comment on specific curves with high 
Miyaoka-Kobayashi number.

Let ${\rm Der}(S) = \{\partial:= a \cdot \partial_{x} +b  \cdot \partial_{y}+c \cdot \partial_{z}, a,b,c \in  S\}$ be the free $S$-module of ${\Bbb{C}}$-linear derivations of the polynomial ring $S$. The graded  $S$-module of derivations for a reduced curve $C : F= 0$ preserving the ideal $\langle F \rangle$ is denoted by ${\rm D}(F) = \{ \partial \in {\rm Der}(S)  :  \partial F \in \langle F \rangle \}$. If $\delta_{E}= x\partial_{x} +y\partial_{y}+z\partial_{z}$ is the Euler derivation and ${\rm D}_{0}(F) =  \{ \partial \in {\rm Der}(S) : \partial F = 0\}$ is the graded $S$-module of Jacobian syzygies of $F$, then we have  ${\rm D}(F) = {\rm D}_{0}(F) \oplus S\cdot \delta_{E}$.

 \begin{defn} \cite{dim24a} We say that a reduced curve $C : F= 0$ is free if ${\rm D}(F)$, or equivalently ${\rm D}_{0}(F)$, is a free graded $S$-module. The exponents $(d_{1},d_{2})$ of a free curve $C$ are the degrees of a basis for the graded $S$-module ${\rm D}_{0}(F)$ with rank 2.
 \end{defn} 
  
  If $C$ is a curve of even
degree $d$ with only simple singularities then the “number”
of simple singularities (by counting each simple singularity
$X_{k}, X = A,D, E$ by its index $k$) of $C$ is less or equal $\frac{3}{4} d^2-\frac{3}{2} d + 1$ \cite{hir82}. 
Now we use these results to construct algebraic curves with many simple singularities of types $A,D$ and $E$  \cite{esc14,esc17}. The singularities have local equations \cite{arn85} 
$$A_{k}: x^{k+1}+y^2=0; D_{k}: x(x^{k-2}+y^2)=0; E_{6}: x^{3}+y^4=0; E_{7}: x(x^{2}+y^3)=0$$
In \cite{dim23} it is shown that reduced plane curves $C$ of odd degree $d=2m+1$ with only $ADE$ singularities have a total Tjurina number at most $\tau(C)=3m^2+1$. Curves with maximum total Tjurina number of $C$,  denoted by $\tau(C)$, were studied in the 80s  (\cite{hir82},  \cite{per82}). Later studies have motivated the following

    \begin{defn}   \cite{dim24a}
      A curve $C : F = 0$ of degree $d$ having only ADE-singularities is maximising if either $d = 2m$ and $\tau(C) = 3m(m-1) + 1$ or $d = 2m+1$ and $\tau(C) = 3m^2 + 1$.
      \end{defn} 
      
           \begin{figure}[h]
 \includegraphics[width=13pc]{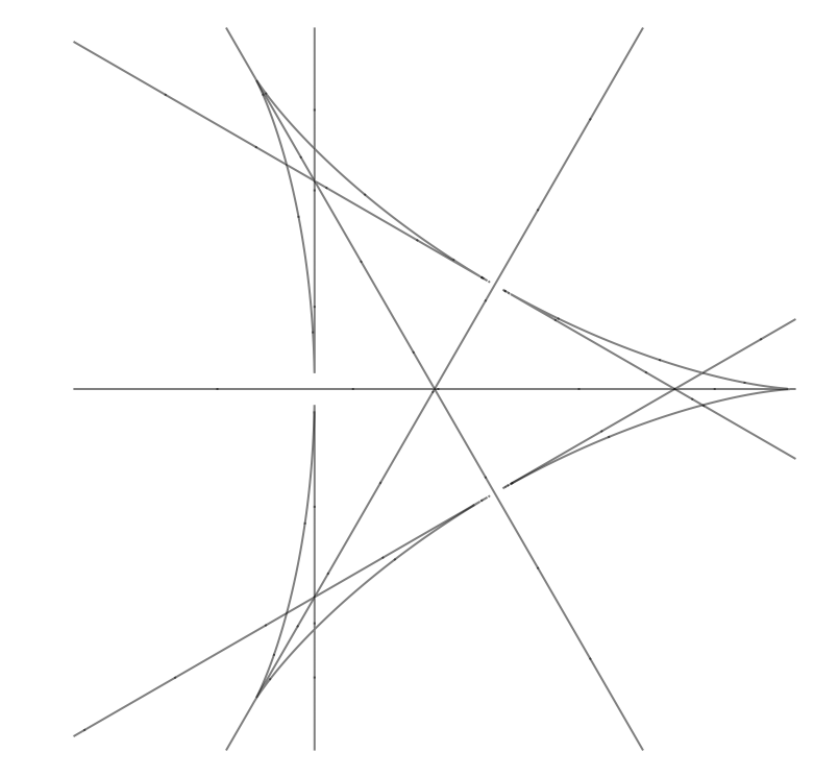}
\caption{\label{label}Maximising  decic $M_{10}$}
\end{figure}
\par
      
      \bigskip\par
Related to the polynomials discussed in Section 2 is the decic curve (see Fig.1 and also Fig.6 when $\mu=k\pi/36, k=30$) 
$$M_{10} := P_{6}(x, y, 5\pi/6) \cdot \delta$$ 
which has $3 E_{7}  \oplus 3 D_{6} \oplus 4 D_{4} \oplus 6A_{1}$ hence it is maximising with $\tau(M_{10})=61$ \cite{esc14}.
\par
  \begin{figure}[h]
 \includegraphics[width=13pc]{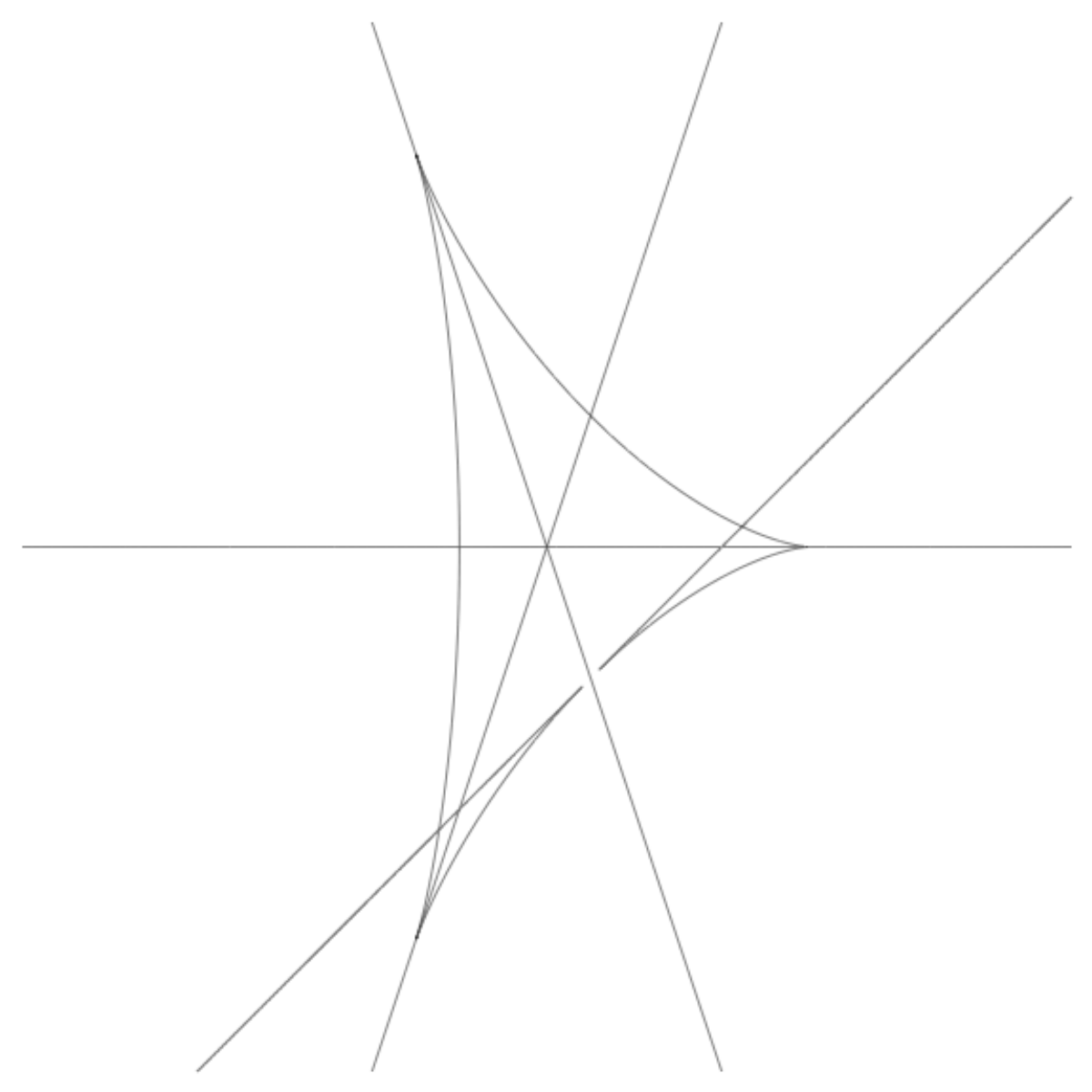}
\caption{\label{label} Maximising octic over the rationals ${\widetilde{M}}_{8}$}
\end{figure}
\par
A maximising octic  $M_{8}$ can be constructed if we subtract from the decic $M_{10}$ two of the three lines tangent to the deltoid ordinary points. The curve is not defined over the rationals, but we can get one defined over the rationals which is given by ${\widetilde{M}}_{8}(x,y):= M_{8}(x,\frac{y}{\sqrt{3}})$ (Fig.2) and has equation $${\widetilde{M}}_{8}(x,y) = (-18 x^2 y + 9 x^3 y - 9 x^2 y^2 + 2 y^3 - x y^3 + y^4) (-27+18 x^2-8 x^3+x^4+6 y^2+8 x y^2+(2 x^2 y^2)/3 + (y^4)/9)$$ 

We can check that  $\tau({\widetilde{M}}_{8})=37$  with the Singular  \cite{gre08} code (short input is used, e.g. $x^4+3x^2$ is denoted $x4+3x2$)
      \bigskip\par
LIB "sing.lib";

ring R= 0, (x,y,z), dp ; 

poly f= (-18x2yz+9x3y-9x2y2+2y3z-xy3+y4)(-27z4+18x2z2-8x3z+x4+6y2z2+8xy2z+(2x2y2)/3 + (y4)/9);

ideal sl= jacob(f);  // the singular locus of f

ideal newsl = groebner(sl); //a groebner basis

mult(newsl); //total tjurina number
  
      \bigskip\par
Another maximising octic $T_{8}$ formed by the union of a deltoid, one bitangent line and three cuspidal tangents was presented in \cite{dim23}. It has $3 E_{7}  \oplus D_{4}  \oplus 2 A_{3} \oplus 6A_{1}$  whereas ${\widetilde{M}}_{8}$ has $3 E_{7}  \oplus D_{6} \oplus D_{4} \oplus 6A_{1}$. Other maximising octics were constructed in  \cite{dim24a} but they also have different types of singularities distribution.
 
 The relation between maximising curves and free curves is characterised by the following 
     
 \begin{thm}  \cite{dim23}
 A curve $C : F= 0$  of degree $d$ having only $ADE$-singularities is maximising if and only if either 
 $d=2m$ and $C$ is a free curve with the exponents $(m-1,m)$ or $d=2m+1$ and $C$ is a free curve with exponents   $(m-1,m+1)$.
 \end{thm}

 As a result of Theorem 3.3 we see that the exponents of the free curves ${\widetilde{M}}_{8}$ and  $M_{10}$ are $(3,4)$ and $(4,5)$ respectively. A maximising nonic (therefore a nonic free curve with exponents (3,5)) has been obtained in \cite{dim24a}, p.23, but maximising curves of odd degree are exceptional \cite{jan25}. In order to study the cases of curves which are free but not maximising other methods must be used \cite{dim17}. Let $R={\Bbb{C}}[x_{0}, x_{1},...,x_{n}]$ be the graded ring of polynomials in $x_{0}, x_{1},...,x_{n}$ with complex coefficients. The vector space of degree $k$ homogeneous polynomials in $R$ is denoted by $R_{k}$. For any $f\in R_{k}$ the graded Milnor or Jacobian algebra is denoted by  $M(f)$.
For a degree $d$ hypersurface $S: f=0$ with isolated singularities in ${\Bbb{P}}^{n}({\Bbb{C}})$  the following integers are defined: 
\par\noindent
1) Coincidence threshold: $ct(S):= {\rm max}\{q: {\rm dim} M(f)_{k}={\rm dim} M(f_{s})_{k}, \forall k\le q\}$, with $f_{s}$ a polynomial in $R_{d}$, such that $S_{s}: f_{s}=0$ is a smooth hypersurface in ${\Bbb{P}}^{n}({\Bbb{C}})$. 
\par\noindent
2) Stability threshold: $st(S):= {\rm min}\{q: {\rm dim} M(f)_{k}=\tau(S), \forall k\ge q\}$. 
\par
Coincidence and Stability thresholds can be used to check if a given curve is free.

 \begin{thm}  \cite{dim17}
 (i) For a reduced free plane curve $C : F= 0$ of degree $d$ one has $ct(F)+st(F)=T$ where $T=3(d-2)$.
 
 (ii) Conversely, suppose that the reduced plane curve $C : F= 0$ of degree $d$ satisfies $ct(F)+st(F) \le T+1$. Then $C $ is free.
 \end{thm}
 
   \begin{figure}[h]
 \includegraphics[width=13pc]{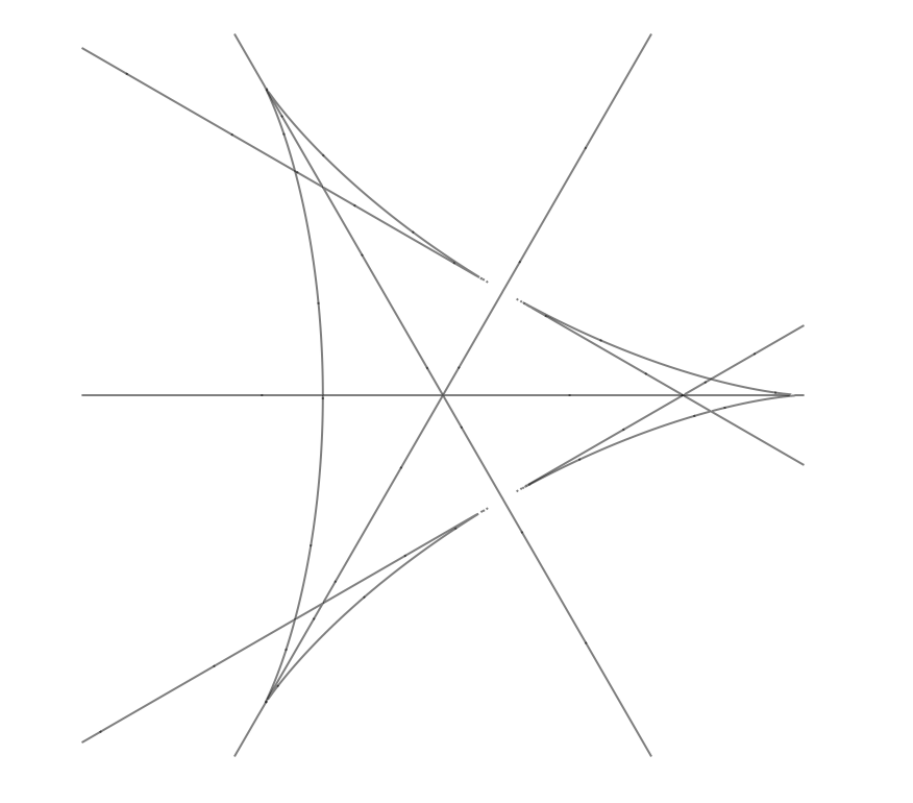}
\caption{\label{label} A free nonic curve $C_{9}$ which is not maximising}
\end{figure}
\par

 If we analyse the singularities of the nonic curve $C_{9} 
 :=(-48 x^2 y+48 x^3 y-12 x^4 y+16 y^3-16 x y^3+40 x^2 y^3-12 y^5)
(-27+18 x^2-8 x^3+x^4+18 y^2+24 x y^2+2 x^2 y^2+y^4)$,
 obtained by subtracting from $M_{10}$  the line $x+1=0$, which is tangent to one of the deltoid ordinary points (Fig.3) we see that  $\tau(C_{9}) = 48 = 3m^2$ therefore it is not maximising. However $C_{9}$ is free  because the computation with Singular shows that  $ ct (C_{9})= 11, st (C_{9})= 10$ and $T=3(d-2)=21=ct (C_{9})+st (C_{9})$. 

\bigskip\par
A family of curves which are significant in the construction of complex ball-quotient surfaces are the Miyaoka-Kobayashi  (MK for short) curves which are certain complex projective plane curves of even degree higher than four, defined and studied by Hirzebruch and Ivinskis. If $Sing(C)$ denotes the set of singularities of the curve $C$ then the MK number is $m(C) :=  \sum_{p \in Sing(C)}m(p)$, where $m(p)$ is the MK number of the singularity $p$  (see Definition 2.1 in \cite{dim24b}). In the case of ADE singularities there are explicit expressions for  $m(p)$ given in Lemma 3.2 in \cite{dim24b}. A MK curve $C$ of even degree $d > 5$ is characterised by $m(C) = \frac{d}{2}(5d -6)$ (Definition 2.4 in \cite{dim24b}).

For degrees higher than seven no example of MK curve is known \cite{dim24b} and the curves presented in this paper are not MK curves. We then consider the problem of finding curves with the maximum $m$.  According to Lemma 3.2 in \cite{dim24b} $m(w)=3(\tau(w)+\epsilon(w))$ for any isolated hypersurface singularity $w$, where 
$\epsilon(A_{k})=\frac{k}{k+1} (k \ge 1), \epsilon(D_{k})=\frac{4k-9}{4k-8} (k \ge 4), \epsilon(E_{6})=\frac{23}{24}, \epsilon(E_{7})=\frac{47}{48}, \epsilon(E_{8})=\frac{119}{120}$ and $\tau(w)$ is the Tjurina number of the singularity. Both maximising octics  ${\widetilde{M}}_{8}$ and $T_{8}$ are not MK curves but $T_{8}$ has an $m(C)/3$ closer to the value  $\tau(C)+\epsilon(C)=\frac{d}{6}(5d -6)$. In fact $m({\widetilde{M}}_{8})<m(T_{8})=135.9375<\frac{d}{2}(5d -6)=136$. The maximising decic $M_{10}$ is not a MK curve but   $m(M_{10})=3(\tau(C)+\epsilon(C))=219.75$ is close to $\frac{d}{2}(5d -6)=220$ and higher than other maximising decics.
\bigskip\par
     
\section{The critical points of $P_{d}(x,y,\mu)$}
\bigskip\par
In Section 5 we will study Hamiltonian systems related to the polynomials $P_{d}(x,y,\mu)$. It will then be necessary to use certain characteristics of the critical points of $P_{d}(x,y,\mu)$.
 \begin{thm} \cite{esc17}
The critical points and critical values of $P_{d}(x,y,\mu)$ have the following properties:
\bigskip\par\noindent
(p1) The critical values are 
$$\zeta=0, \zeta_{M}(\mu)=6{\rm cos}\mu+2{\rm cos}3\mu, \zeta_{m_{1}}(\mu)=\zeta_{M}(\mu-\frac{2\pi }{3}), \zeta_{m_{2}}(\mu)=\zeta_{M}(\mu+\frac{2\pi }{3}).$$ 
\par\noindent
(p2) The critical points with critical values $\zeta_{M}(\mu),\zeta_{m_{1}}(\mu)$ and $\zeta_{m_{2}}(\mu)$ have the same coordinates $\forall \mu \in  {\Bbb{R}}$. 
\par\noindent
(p3) At $\mu \in \Lambda_{\Sigma}:=\{k \frac{\pi}{3}, k \in  {\Bbb{Z}}, 0 \le k \le 5\}$, either $P_{d}(x,y,\mu)$ or $-P_{d}(x,y,\mu)$ have all the maxima with critical value 8 and all the minima with value -1.
\par\noindent
(p4) All the maxima have values $3 \sqrt{3}$ and all the minima $-3 \sqrt{3}$ at $\mu \in \Lambda_{S}:=\{(2k+1) \frac{\pi}{6}, k \in  {\Bbb{Z}}, 0 \le k \le 5 \}$.
\par\noindent
(p5) For each $\mu \in \Lambda_{3}:=[0,2\pi) \setminus \{ k \frac{\pi}{6}, k \in  {\Bbb{Z}}, 0 \le k \le 11 \}$, $P_{d}(x,y,\mu)$ has critical values $\zeta=0$ and $\zeta_{m}, m=1,2,3$, with $0<\left | \zeta_{3} \right | <1< \left | \zeta_{2} \right | <3 \sqrt{3}<\left | \zeta_{1} \right |<8$ and $sgn({\zeta_{1}}) \ne sgn({\zeta_{2}})=sgn({\zeta_{3}})$.

\end{thm}

The critical points of $P_{d}(x,y,\mu)$
are 
\begin{equation}
(x_{c}, y_{c})=({\rm cos} (2\pi (u_{c}+v_{c}))+{\rm cos} (2\pi u_{c})+{\rm cos} (2\pi v_{c}),{\rm sin} (2\pi (u_{c}+v_{c}))-{\rm sin} (2\pi u_{c})-{\rm sin} (2\pi v_{c}))
\end{equation}
 and a direct computation of the critical values leads to the following expressions for $u_{c},v_{c}$ ($k,l \in {\Bbb{Z}}$) \cite{esc17}:
  \bigskip\par 
     \par\noindent
 (a) $\zeta_{M}(\mu)=6{\rm cos}\mu+2{\rm cos}3\mu$; $u_{c}=\frac{3k+1}{3d},v_{c}=\frac{3l+1}{3d}$.
  \bigskip\par
     \par\noindent
 (b) $\zeta_{m_{1}}(\mu)=6{\rm cos}(\mu-\frac{2\pi }{3})+2{\rm cos}3\mu$; $u_{c}=\frac{3k+2}{3d},v_{c}=\frac{3l+2}{3d}$.
    \bigskip\par 
   \par\noindent
 (c) $\zeta_{m_{2}}(\mu)=6{\rm cos}(\mu+\frac{2\pi }{3})+2{\rm cos}3\mu$; $u_{c}=\frac{k}{d},v_{c}=\frac{l}{d}$.
     \bigskip\par 
   \par\noindent
 (d) $\zeta=0$; 
   \par
  (d1) $u_{c}=\frac{6k-1}{6d}-\frac{\mu}{\pi d}, v_{c}=\frac{6l-1}{6d}-\frac{\mu}{\pi d}$;
     \par
  (d2) $u_{c}=\frac{6k-1}{6d}-\frac{\mu}{\pi d}, v_{c}=\frac{3l+1}{3d}+\frac{2\mu}{\pi d}$;
       \par
  (d3) $u_{c}=\frac{3k+1}{3d}+\frac{2\mu}{\pi d}, v_{c}=\frac{6l-1}{6d}-\frac{\mu}{\pi d}$.
\bigskip\par
  The points $(u_{c},v_{c})$ in the $(u,v)$-plane are situated in the fundamental region, denoted by $W$, of the affine Weyl group of $\bold{A}_{2}$. They satisfy  $\max \{-2v,v\}<u<\frac{1-v}{2}$  \cite{esc16} when $\mu \notin \Lambda_{S}$ whereas when $\mu \in \Lambda_{S}$ (case (p4) in Theorem 4.1) some points are in the border of $W$. This allows to compute the possible values of $k,l$ and hence the number of critical points of $P_{d}(x,y,\mu)$, which are situated  inside ${\rm Int} \delta$ in the $(x,y)$-plane, except when $\mu \in \Lambda_{S}$  where some critical points with $\zeta=0$ lie on $\delta$ (\cite{esc17}, Figs.2,3).
\par
\section{Bifurcations on the polynomial Hamiltonian systems associated with $P_{d}(x,y,\mu)$}

\bigskip\par
We now study one-parameter families of autonomous ordinary differential equations  $\dot {\bf r}=f({\bf r},\mu)$, which generate flows, where ${\bf r}(t)=(x(t),y(t)) \in \Bbb{R}^2$ are the dependent variables, $t \in \Bbb{R}$ is the independent variable and $\mu \in \Bbb{R}$ is the parameter. The simplest solutions are constants corresponding to $f({\bf r},\mu)=0$, and are called stationary points or critical points. Bifurcations or catastrophes deal with the change in the number of solutions of $f({\bf r},\mu)=0$ as the parameter $\mu$ varies, and then can be treated as part of singularity theory \cite{arn72, arn85}. They were motivated to some extent by some questions in biology \cite{toh76} and have been applied to a variety of problems in science and engineering  \cite{gil81}.  
\par
The defining differential equations we consider can be derived from a polynomial Hamiltonian $H(x,y)$ for which we have the following result (\cite{lay15}, p.298): 
 \begin{thm}
Any nondegenerate critical point of an analytic Hamiltonian system 
\begin{equation}
\dot x=\partial_{y}H(x,y), \dot y=-\partial_{x}H(x,y)
\end{equation}
is either a saddle or a center; again $(x_{0},y_{0})$ is a saddle for Eq.(5.1) if and only if it is a saddle of the Hamiltonian function $H(x,y)$ and a strict local maximum or minimum of the function $H(x,y)$ is a center for Eq.(5.1).
\end{thm}
We analyse codimension-one bifurcations \cite{kus95}  for a type of dynamical systems connected with  $P_{d}(x,y,\mu)$. We treat the Hamiltonian system
\begin{equation}
\dot x=\partial_{y}P_{d}(x,y,\mu), \dot y=-\partial_{x}P_{d}(x,y,\mu)
\end{equation}
\par
Although the results we obtain are valid for the system given by Eq.(5.2) one can also consider a related gradient system $\dot {\bf x}=-{\rm grad} F({\bf x})$, where ${\rm grad} F=(\partial_{x}F, \partial_{y}F)^T$ (\cite{lay15}, p.302): 
 \begin{thm}
The planar system given by $\dot x=f(x,y),  \dot y=g(x,y)$ is a Hamiltonian system, if and only if the system orthogonal to it, given by $\dot x=g(x,y),  \dot y=-f(x,y)$ is a gradient system.
\end{thm}
\par
The critical or stationary points of the orthogonal gradient system 
\begin{equation}
\dot x=-\partial_{x}P_{d}(x,y,\mu), \dot y=-\partial_{y}P_{d}(x,y,\mu)
\end{equation}
are the same as those of the system given by Eq.(5.2). The centers of a planar system correspond to the nodes of its orthogonal system and the saddles (foci) of the planar system are the saddles (foci) of its orthogonal system \cite{lay15}. Also at the regular points the trajectories of these systems are orthogonal. 
\par

The bifurcation or catastrophe sets of the dynamical system defined by Eq.(5.2) are analysed in this Section by taking into account Theorem 4.1. 
\par
 \begin{thm} The polynomial Hamiltonian system 
$$
\dot x=\partial_{y}P_{d}(x,y,\mu)$$
$$ \dot y=-\partial_{x}P_{d}(x,y,\mu)
$$
has 
       \par
          \par\noindent
(a) two bifurcation points at $\mu=\frac{\pi}{6}, \frac{\pi}{2}$ for $d \neq 3n$
       \par
          \par\noindent
(b) six  bifurcation points at $\mu=(2l+1)\frac{\pi}{6}$, $l=0,1,...5$ for  $d= 3n$.
 \end{thm}
  \begin{proof}
We consider both the degenerate and the nondegenerate critical points of $P_{d}(x,y,\mu)$. The bifurcations occur when degenerate critical points appear. In the proof of Theorem 4.1 (Theorem 2.2 in  \cite{esc17}) there is a description of the type of line arrangements associated with the level curve $P_{d}(x,y,\mu)=0$ in terms of the behaviour of $\zeta_{M}(\mu), \zeta_{m_{1}}(\mu), \zeta_{m_{2}}(\mu)$ for $\mu \in [0,2\pi)=\Lambda_{\Sigma}\cup \Lambda_{S}\cup \Lambda_{3}$ (\cite{esc17}, Fig.1):  
\par
(1) $\mu \in \Lambda_{\Sigma}$: one of the functions $\zeta$ has a maximum with critical value $8$ and the other two have value $-1$, or a minimum with critical value $-8$ and the other two $1$. The critical points of $P_{d}(x,y,\mu)$ are nondegenerate and the curves $P_{d}(x,y,\mu)=0$ are simple arrangements of lines denoted by $\Sigma^{d}_{D}$ ($\forall d$) and $\Sigma^{d}_{C}$ (only for $d=3n$). 
\par
(2) $\mu \in \Lambda_{S}$: one $\zeta$ has an inflection point with critical value $0$ and the others have values $-3\sqrt{3}$ and $3\sqrt{3}$, the curves $P_{d}(x,y,\mu)=0$ being simplicial arrangements denoted by  $S^{d}_{D}$ ($\forall d$) or $S^{d}_{C}$ (only for $d=3n$) hence we find degenerate critical points in those cases. 
\par
(3) $\mu \in \Lambda_{3}$: the curves $P_{d}(x,y,\mu)=0$ are simple arrangements not of type $\Sigma^{d}_{C}$ and $\Sigma^{d}_{D}$ but they have the same number of critical points with critical value $0$, namely $d \choose 2$, and the other critical values transform as described in what follows. 
 \bigskip\par 
There are two types of critical or stationary points of the associated Hamiltonian system: saddles in the vertices of the line arrangements, and centers at the interior of the triangles or the closed non-triangular cells (quadrilateral, pentagonal or hexagonal cells).
We now analyse the evolution of the critical points of $P_{d}(x,y,\mu)$ as $\mu$ varies for the cases (a) $d \neq 3n$ and (b) $d=3n$, which are qualitatively different:
 \bigskip\par
    \par\noindent
(a) $d \neq 3n$
 \bigskip\par
 (a1) $0 \le \mu \le \frac{\pi}{6}$. The critical values of $P_{d}(x,y,0)$, associated with the simple line arrangements $\Sigma^{d}_{D}$ with dihedral symmetry, have been analysed in (Lemma 1, \cite{esc16}). There are $\frac{(d-1)(d-2)}{6}$ local maxima with $\zeta_{M}=8$ inside the non triangular zones. The $\frac{(d-1)(d-2)}{6}$ local minima with $\zeta_{m_{1}}=-1$ and the $\frac{(d-1)(d-2)}{6}$ local minima with $\zeta_{m_{2}}=-1$ are located inside the triangular ones. As $\mu$ increases its value within the interval $[0,\frac{\pi}{6}]$ the curves $P_{d}(x,y,\mu)=0$ are still simple arrangements (\cite{esc17}, Fig.2, for $d=5$), the triangular zones containing the critical points with $\zeta_{m_{1}}$ decrease in size and, for each triangle (half of the total), its three vertices, which are $A_{1}$ singularities of the curve $P_{d}(x,y,\mu)=0$ corresponding to $\zeta=0$, collapse into one $D_{4}$ singularity when $\mu=\frac{\pi}{6}$: $\zeta_{m_{1}}(\frac{\pi}{6})=0$ and the values $(u_{c},v_{c})$ of the critical points with $\zeta(\frac{\pi}{6})=0$ are (d1) $(\frac{3k-1}{3d},\frac{3l-1}{3d})$, (d2) $(\frac{3k-1}{3d},\frac{3l+2}{3d})$, (d3) $(\frac{3k+2}{3d},\frac{3l-1}{3d})$, which, by appropriate relabelling, coincide with those linked to $\zeta_{m_{1}}$ which are always in ${\rm Int} \delta$. The remaining critical points with $\zeta(\frac{\pi}{6})=0$ are on $\delta$ (\cite{esc17}, Fig.2(d), for $d=5$) and if we consider the union of $\delta$ and the simplicial arrangement $S^{d}_{D}$ then they correspond to $d-1$ $D_{6}$ singularities \cite{esc21}. 
    \par
(a2) $\frac{\pi}{6}<\mu \le \frac{\pi}{3}$. Each non-hyperbolic point of $P_{d}(x,y, \frac{\pi}{6})$, which is a $D_{4}$ singularity of the curve $P_{d}(x,y, \frac{\pi}{6})=0$, splits into 3 saddles and 1 center of the Hamiltonian system. We have local minima inside the $\frac{(d-1)(d-2)}{6}$ non triangular zones containing the critical value $\zeta_{m_{2}}$ which decreases from $\zeta_{m_{2}}(\frac{\pi}{6})=-3 \sqrt{3}$ until $\zeta_{m_{2}}(\frac{\pi}{3})=-8$. The local maxima are inside the $\frac{(d-1)(d-2)}{6}$ triangular zones containing $\zeta_{M}$ decreasing from $\zeta_{M}(\frac{\pi}{6})=3 \sqrt{3}$ until $\zeta_{M}(\frac{\pi}{3})=1$, and inside the $\frac{(d-1)(d-2)}{6}$ triangular zones containing $\zeta_{m_{1}}$ increasing from $\zeta_{m_{1}}(\frac{\pi}{6})=0$ until $\zeta_{m_{1}}(\frac{\pi}{3})=1$. The curve at $\mu = \frac{\pi}{3}$ is, as for $\mu = 0$, a simple arrangement $\Sigma^{d}_{D}$. The first bifurcation therefore occurs when $\mu=\frac{\pi}{6}$.
   \par
 (a3) $\frac{\pi}{3}<\mu  \le \frac{\pi}{2}$. The  vertices of the $\frac{(d-1)(d-2)}{6}$ triangles containing local maxima with a critical value $\zeta_{M}$ in the interior collapse into $D_{4}$ singularities when $\mu=\frac{\pi}{2}$ analogously to the cases for $\zeta_{m_{1}}$ in (a1): $\zeta_{M}(\frac{\pi}{2})=0$ and the values $(u_{c},v_{c})$ of the critical points with $\zeta(\frac{\pi}{2})=0$ are (d1) $(\frac{3k-2}{3d},\frac{3l-2}{3d})$, (d2) $(\frac{3k-2}{3d},\frac{3l+4}{3d})$, (d3) $(\frac{3k+4}{3d},\frac{3l-2}{3d})$, which can be rewritten as $(\frac{3k+1}{3d},\frac{3l+1}{3d})$. The other $d-1$ critical points with $\zeta(\frac{\pi}{2})=0$ are in $\delta$ and remain $A_{1}$ singularites of the planar curve.
 \par
 (a4) $\frac{\pi}{2}<\mu  \le \frac{2\pi}{3}$. The second bifurcation takes place when $\mu=\frac{\pi}{2}$. The curves and trajectories for $\mu=\frac{2\pi}{3}$ are, up to a rotation, the same as those for $\mu=0$ and there are no more distinct bifurcations for $\mu>\frac{\pi}{2}$.  
 \bigskip\par 
    \par\noindent
(b) $d =3n$
 \bigskip\par
 (b1) $0 \le \mu \le \frac{\pi}{6}$. The curves $P_{d}(x,y, 0)=0$ are formed by the simple line arrangements $\Sigma^{d}_{C}$ with cyclic symmetry. There are $\frac{d(d-3)}{6}$ local maxima with $\zeta_{M}=8$ inside the non triangular zones. Located at the triangular zones there are $\frac{d(d-3)}{6}$ local minima with $\zeta_{m_{1}}=-1$ and $1+\frac{d(d-3)}{6}$ local minima with $\zeta_{m_{2}}=-1$ (Lemma 1, \cite{esc16}). In the interval $[0,\frac{\pi}{6})$  the curves are still simple arrangements (\cite{esc17}, Fig.3, for $d=6$). The triangular zones containing the critical points with critical value $\zeta_{m_{1}}$ shrink as $\mu$ increases its value and, for each triangle, the three $A_{1}$ singularities of the curve situated at its vertices, which correspond to $\zeta=0$, collapse into one $D_{4}$ singularity when $\mu=\frac{\pi}{6}$, because $\zeta_{m_{1}}(\frac{\pi}{6})=0$ and the values $(u_{c},v_{c})$ of the critical points with $\zeta(\frac{\pi}{6})=0$ are (d1) $(\frac{3k-1}{3d},\frac{3l-1}{3d})$, (d2) $(\frac{3k-1}{3d},\frac{3l+2}{3d})$, (d3) $(\frac{3k+2}{3d},\frac{3l-1}{3d})$, which, by appropriate relabelling, coincide with those corresponding to $\zeta_{m_{1}}$. The remaining critical points with $\zeta(\frac{\pi}{6})=0$ are in $\delta$ and correspond  to $d$ $D_{6}$ singularities of $\delta \cup S^{d}_{C}$  \cite{esc21}. 
    \par
 (b2) $\frac{\pi}{6}<\mu  \le \frac{\pi}{3}$. Each one of the $d$ non-hyperbolic points of $P_{d}(x,y, \frac{\pi}{6})$ splits into 3 saddles and 1 center of Eq.(5.2). The local minima are inside the $1+\frac{d(d-3)}{6}$ non triangular zones with critical value $\zeta_{m_{2}}$ which decreases from $\zeta_{m_{2}}(\frac{\pi}{6})=-3 \sqrt{3}$ until $\zeta_{m_{2}}(\frac{\pi}{3})=-8$. There are local maxima inside the $\frac{d(d-3)}{6}$ triangular zones with critical value $\zeta_{M}$, which decreases from $\zeta_{M}(\frac{\pi}{6})=3 \sqrt{3}$ until $\zeta_{M}(\frac{\pi}{3})=1$, and there are $\frac{d(d-3)}{6}$ triangular zones containing local maxima with critical value $\zeta_{m_{1}}$ which increases from $\zeta_{m_{1}}(\frac{\pi}{6})=0$ until $\zeta_{m_{1}}(\frac{\pi}{3})=1$. The curve at $\mu = \frac{\pi}{3}$ is, in contrast to $\mu = 0$, a simple arrangement $\Sigma^{d}_{D}$. As in (a) the first bifurcation is found when $\mu=\frac{\pi}{6}$.
   \par
 (b3) $\frac{\pi}{3}<\mu  \le \frac{\pi}{2}$. When $\mu=\frac{\pi}{2}$ the  vertices of the $\frac{d(d-3)}{6}$ triangles containing in the interior  local maxima with a critical value $\zeta_{M}$ produce $D_{4}$ singularities: $\zeta_{M}(\frac{\pi}{2})=0$ and the values $(u_{c},v_{c})$ of the critical points with $\zeta(\frac{\pi}{2})=0$ are (d1) $(\frac{3k-2}{3d},\frac{3l-2}{3d})$, (d2) $(\frac{3k-2}{3d},\frac{3l+4}{3d})$, (d3) $(\frac{3k+4}{3d},\frac{3l-2}{3d})$, which can be rewritten as $(\frac{3k+1}{3d},\frac{3l+1}{3d})$. In the deltoid curve there are $d$ critical points with critical value $\zeta(\frac{\pi}{2})=0$ that are $A_{1}$ singularites of $P_{d}(x,y, \frac{\pi}{2})=0$ which is a simplicial arrangement $S^{d}_{C}$ as in the case $\mu=\frac{\pi}{6}$.
 \par
 (b4) $\frac{\pi}{2}<\mu  \le \frac{2\pi}{3}$. Also as in (a) the second bifurcation takes place when $\mu=\frac{\pi}{2}$.  Each $D_{4}$ singularity of $P_{d}(x,y, \frac{\pi}{2})=0$ splits into 3 saddles and 1 center of Eq.(5.2). We have local maxima inside the $\frac{d(d-3)}{6}$ non triangular zones with the critical value $\zeta_{m_{1}}$ increasing from $\zeta_{m_{1}}(\frac{\pi}{2})=3 \sqrt{3}$ until $\zeta_{m_{1}}(\frac{2\pi}{3})=8$, local minima inside the $\frac{d(d-3)}{6}$ triangular zones with $\zeta_{M}$ decreasing from  $\zeta_{M}(\frac{\pi}{2})=0$ until $\zeta_{M}(\frac{2\pi}{3})=-1$, and local minima inside the $1+\frac{d(d-3)}{6}$ triangular zones with $\zeta_{m_{2}}$ increasing from $\zeta_{m_{2}}(\frac{\pi}{2})=-3 \sqrt{3}$ until $\zeta_{m_{2}}(\frac{2\pi}{3})=-1$. The level curves for $\mu=\frac{2\pi}{3}$ are, up to a rotation and reflection, the same as those for $\mu=0$ but the trajectories have opposite orientations.
     \par
 (b5) $\frac{2\pi}{3}<\mu  \le \frac{5\pi}{6}$. As $\mu$ increases its value within the interval $(\frac{2\pi}{3},\frac{5\pi}{6}]$  the curves are still simple arrangements, the $1+\frac{d(d-3)}{6}$ triangular zones containing the critical points with $\zeta_{m_{2}}$ decrease in size and give one $D_{4}$ singularity when $\mu=\frac{5 \pi}{6}$: $\zeta_{m_{2}}(\frac{5\pi}{6})=0$ and the values $(u_{c},v_{c})$ of the critical points with $\zeta(\frac{5\pi}{6})=0$ are (d1) $(\frac{k-1}{d},\frac{l-1}{d})$, (d2) $(\frac{k-1}{d},\frac{l+2}{d})$, (d3) $(\frac{k+2}{d},\frac{l-1}{d})$, which, by appropriate relabelling, coincide with those associated with $\zeta_{m_{2}}$. The remaining critical points with $\zeta(\frac{5\pi}{6})=0$ are in $\delta$  and correspond to the $d$ $D_{6}$ singularities of $\delta \cup S^{d}_{D}$ \cite{esc21}. 
      \par
 (b6) $\frac{5\pi}{6}<\mu  \le \pi$. The third bifurcation occurs when $\mu=\frac{5\pi}{6}$. Each non-hyperbolic point of $P_{d}(x,y, \frac{5\pi}{6})$ splits into 3 saddles and 1 center of the Hamiltonian system. We have local minima inside the $\frac{d(d-3)}{6}$ non triangular zones with  $\zeta_{M}$ decreasing from $\zeta_{M}( \frac{5\pi}{6})=-3 \sqrt{3}$ until $\zeta_{M}(\pi)=-8$, local maxima inside the $1+\frac{d(d-3)}{6}$ triangular zones with $\zeta_{m_{2}}$ increasing from $\zeta_{m_{2}}( \frac{5\pi}{6})=0$ until $\zeta_{m_{2}}(\pi)=1$, and the $\frac{d(d-3)}{6}$ triangular zones contain local maxima with $\zeta_{m_{1}}$ decreasing from $\zeta_{m_{1}}( \frac{5\pi}{6})=3 \sqrt{3}$ until $\zeta_{m_{1}}(\pi)=1$. The level curves for $\mu=\pi$ are the same as those for $\mu=0$ but the trajectories have opposite orientations.
      \par
 (b7) $\pi<\mu  \le 2\pi$. The level curves are the same as $0<\mu  \le \pi$ but the trajectories have also opposite orientations. We then find three additional bifurcation points at $\mu=\frac{7\pi}{6}, \frac{9\pi}{6}, \frac{11\pi}{6}$.

\end{proof}
\par
In Figs.4,5 we have represented the trajectories of the Hamiltonian system corresponding to $P_{5}(x,y,\mu)$. The level curve $P_{5}(x,y,\mu)=0$ is drawn on the left and the trajectories of the Hamiltonian flow on the right. There are two bifurcation points at $\mu=\frac{k\pi}{30}$, $k=5,15$. The case $k=20$ is topologically equivalent to $k=0$: they are related by a rotation. The non-hyperbolic points are the result of the collapsing of 3 saddles and 1 center. Both bifurcations consist in the splitting of the non-hyperbolic points ($D_{4}$-singularities of the level curves) into 3 saddles ($A_{1}$-singularities of the level curves) and 1 center in the interior of the triangle whose vertices are the 3 saddles. 
\par     
The trajectories for $d=6$ can be seen in Figs.6,7 where we have represented both the level curves and the orbits in the same coordinate frames. We can see some qualitative differences in relation to the case $d=5$, among them the appearance of two types of both simple and simplicial line arrangements with dihedral and cyclic symmetries. The parameter values of the bifurcation points are $\mu=\frac{k\pi}{36}$, $k=6,18,30,42,54,66$, which correspond to all $\mu \in \Lambda_{S}$.
\par
  \begin{figure}[h]
 \includegraphics[width=32pc]{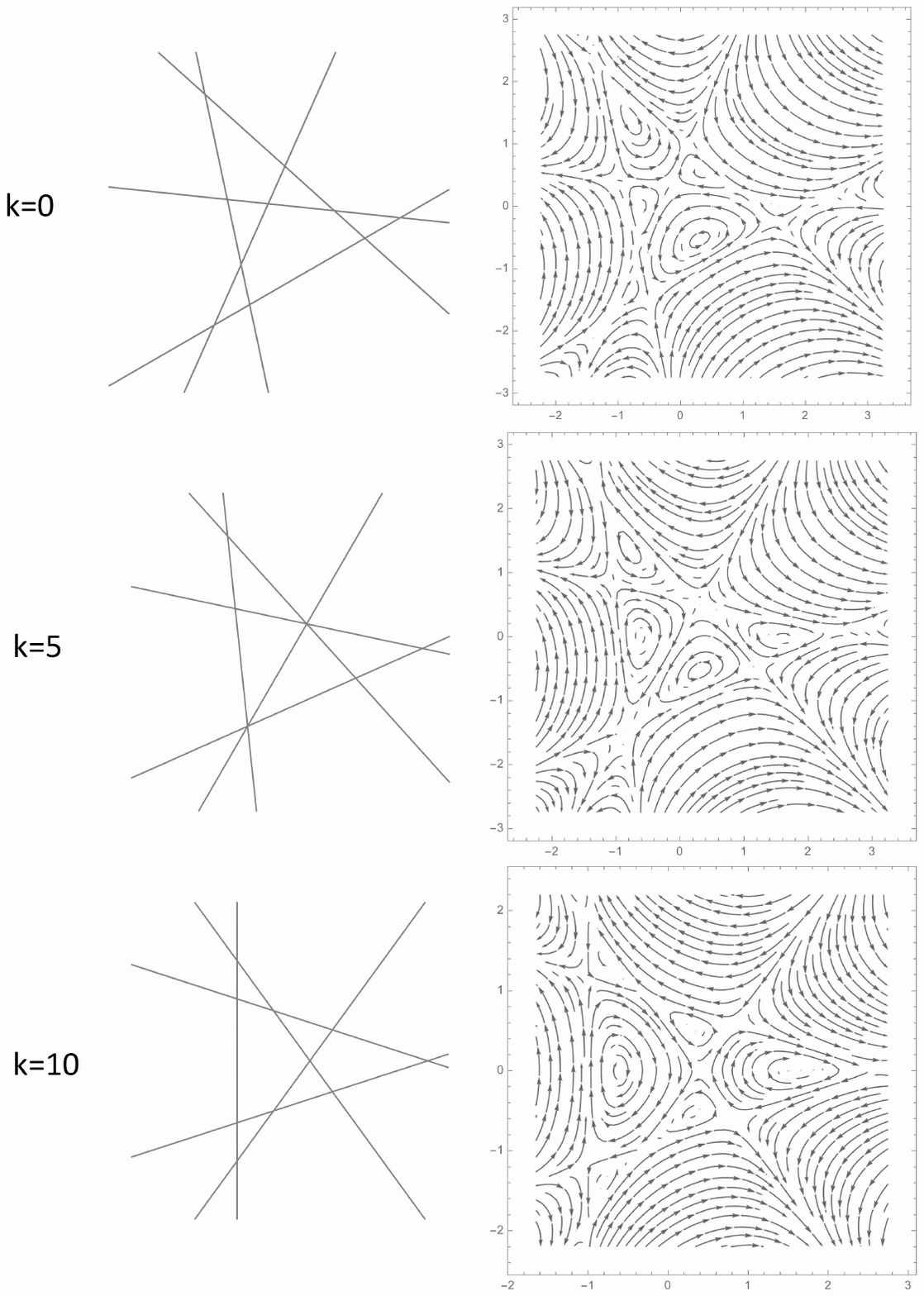}
\caption{\label{label}Level curves $P_{5}(x,y,\mu)=0$ (left) and trajectories (right) for  $\mu=k\pi/30$, $k=0,5,10$. Local minima (maxima) of $P_{5}(x,y,\mu)$ correspond to centers with trajectories having clockwise (anticlockwise) orientation.}
\end{figure}
\par
  \begin{figure}[h]
 \includegraphics[width=30pc]{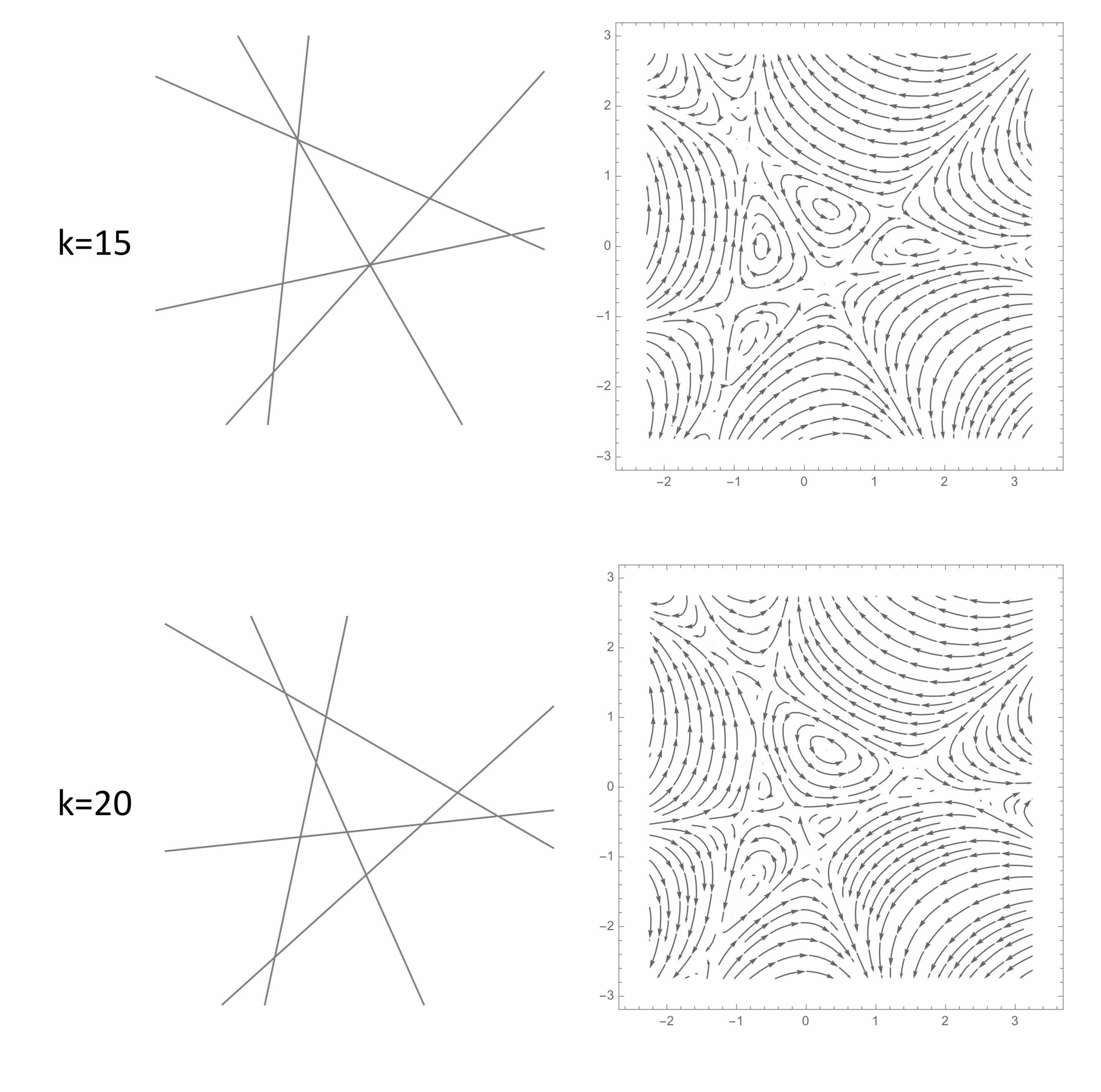}
\caption{\label{label}Level curves $P_{5}(x,y,\mu)=0$ (left) and trajectories (right) for $\mu=k\pi/30$, $k=15,20$.}
\end{figure}
\par
  \begin{figure}[h]
 \includegraphics[width=30pc]{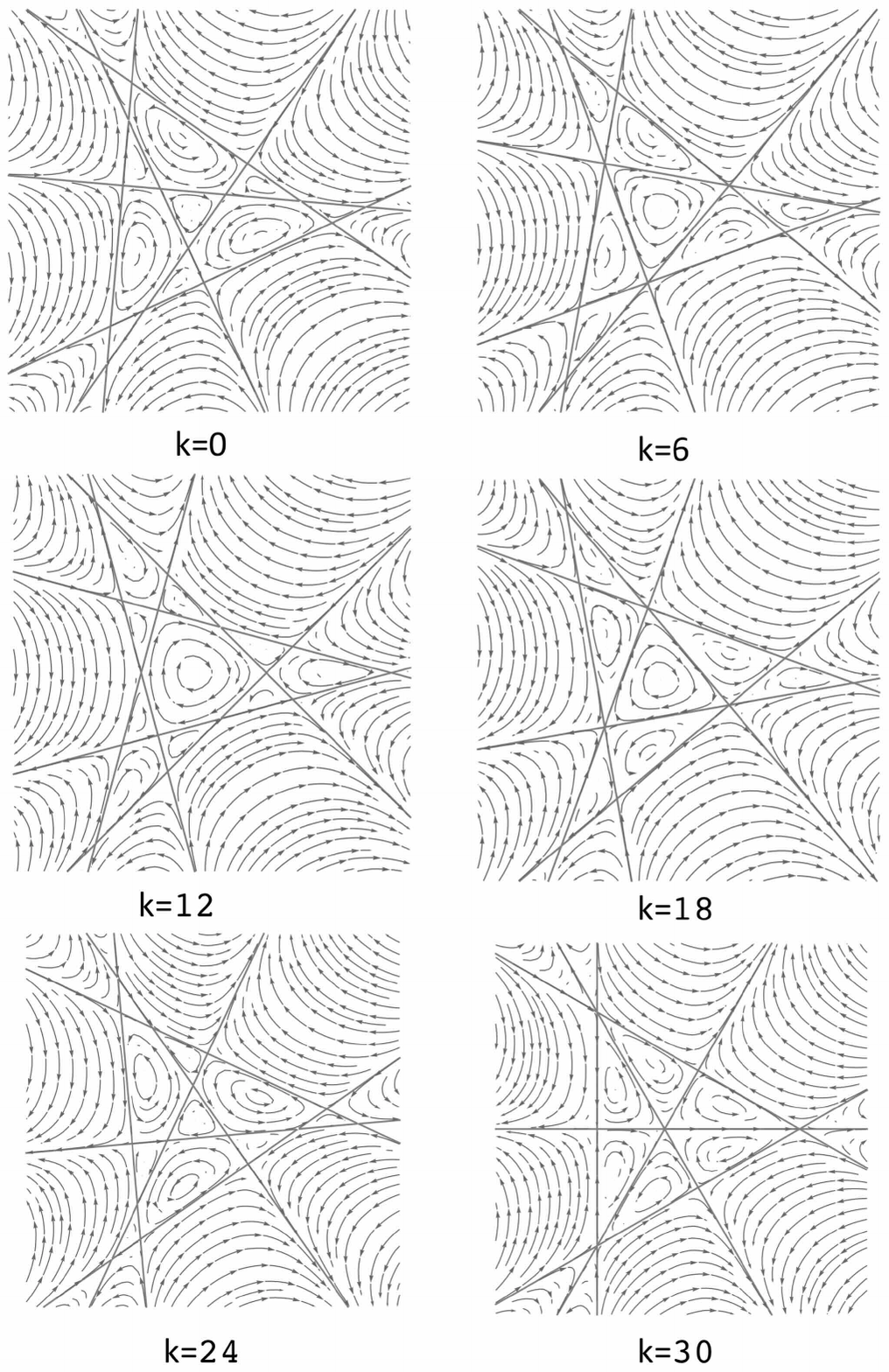}
\caption{\label{label}Level curves $P_{6}(x,y,\mu)=0$ and trajectories for  $\mu=k\pi/36$, $k=0,6,12,18,24,30$.}
\end{figure}
\par
  \begin{figure}[h]
 \includegraphics[width=30pc]{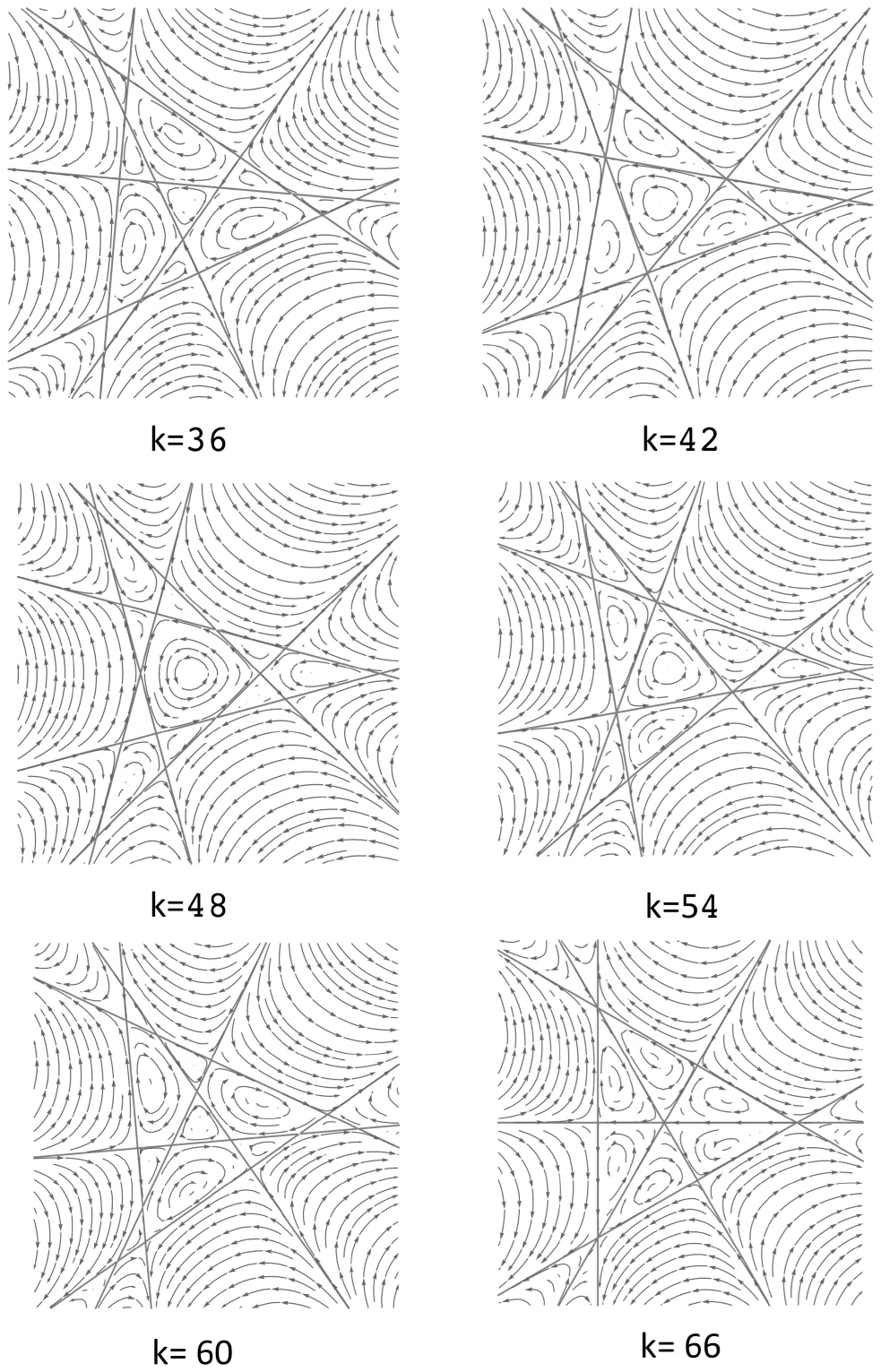}
\caption{\label{label}Level curves $P_{6}(x,y,\mu)=0$ and trajectories for $\mu=k\pi/36$, $k=36,42,48,54,60,66$. The curves are the same as in Fig.6 for $k=0,6,12,18,24,30$ respectively, but the trajectories have reversed orientations.}
\end{figure}

\end{document}